\documentclass{comnet-mod}

\usepackage{cleveref}
\usepackage{makecell}
\usepackage{mathtools}
\usepackage{comment}
\usepackage{subfig}
\usepackage{graphicx}
\usepackage{float}
\usepackage{booktabs}
\usepackage[export]{adjustbox}

\usepackage[colorinlistoftodos,bordercolor=orange,backgroundcolor=orange!20,linecolor=orange,textsize=scriptsize]{todonotes}

\title{Epidemic Thresholds of Infectious Diseases on Tie-Decay Networks}

\shorttitle{Epidemic Thresholds of Infectious Diseases on Tie-Decay Networks} 
\shortauthorlist{Q. Chen and M. A. Porter} 

\author{
\name{Qinyi Chen}
\address{Operations Research Center, Massachusetts Institute of Technology, Cambridge, MA 02139, USA}
\address{Department of Mathematics, University of California, Los Angeles, CA 90095, USA}
\email{qinyic@mit.edu}
\name{Mason A. Porter$^*$}
\address{Department of Mathematics, University of California, Los Angeles, CA 90095, USA}
\address{Santa Fe Institute, Santa Fe, New Mexico, 87501, USA}
\email{$^*$Corresponding author: mason@math.ucla.edu}
}

\begin{document}

\maketitle



\begin{abstract}
{In the study of infectious diseases on networks, researchers calculate epidemic thresholds to help forecast whether a disease will eventually infect a large fraction of a population. Because network structure typically changes in time, which fundamentally influences the dynamics of spreading processes on them and in turn affects epidemic thresholds for disease propagation, it is important to examine epidemic thresholds in temporal networks. Most existing studies of epidemic thresholds in temporal networks have focused on models in discrete time, but most real-world networked systems evolve continuously in time. In our work, we encode the continuous time-dependence of networks into the evaluation of the epidemic threshold of a susceptible--infected--susceptible (SIS) process by studying an SIS model on tie-decay networks. We derive the epidemic-threshold condition of this model, and we perform numerical experiments to verify it. We also examine how different factors---the decay coefficients of the tie strengths in a network, the frequency of interactions between nodes, and the sparsity of the underlying social network in which interactions occur---lead to decreases or increases of the critical values of the threshold and hence contribute to facilitating or impeding the spread of a disease. We thereby demonstrate how the features of tie-decay networks alter the outcome of disease spread.}
{Temporal networks, tie-decay networks, epidemic thresholds, network epidemiology}
\end{abstract}



\section{Introduction}
\label{sec:introduction}

Infectious diseases spread over social networks, and there is thus much research on the spread of diseases on networks \cite{Pastor2015,kiss2017,Newman2018}. The simplest type of network is a graph, in which each node represents an entity (e.g., an individual who is prone to infection) and each edge represents a tie (such as a social relationship) between two entities. Disease transmission occurs across edges. Each node has an associated state---such as susceptible, infected, recovered, zombified, or something else---and different states are appropriate for different diseases. Each state is called a ``compartment'', and models of infectious diseases with such compartments are called ``compartmental models'' \cite{Brauer2019}. Common compartmental models of infectious diseases include susceptible--infected--susceptible (SIS) processes, susceptible--infected--recovered (SIR) processes, and susceptible--exposed--infected--recovered (SEIR) processes. By modeling the contact patterns of a set of individuals using a network, one can examine the spread of an infectious disease on a social network. This, in turn, helps improve forecasts of disease spread. For example, researchers have used network models to study the spread of COVID-19 \cite{covid-review,arino2021describing}. Such work has important policy implications \cite{Arenas2020,Herrmann2020}. 

Many studies of the spread of infectious diseases on social networks aim to determine whether a disease will die out or spread to a large fraction of a population. To do this, scholars often try to calculate an \emph{epidemic threshold} to give a condition that characterizes whether a disease eventually leads to a large outbreak in a population \cite{kiss2017,Pastor2015}. The critical value of an epidemic threshold depends on the choice of compartment model, the values of the parameters of the model, and the structure of the network on which a disease is spreading. There are several theoretical approaches for estimating the epidemic threshold of a model of disease spread on a network. For example, one can use a heterogeneous mean-field theory \cite{Pastor2001}, a quenched mean-field theory \cite{Gomez2010, castellano2010}, or a dynamic message-passing theory \cite{Karrer2010}. These three approaches tend to work well for forecasting the outcome of the spread of a disease on a large and sparse network \cite{Wang2016}, and they have been used to study how various factors (e.g., degree--degree correlations \cite{Boguna2003} and clustering \cite{Serrano2006}) can affect an epidemic threshold. 

Early research on epidemic thresholds focused on time-independent contact networks with specific topological structures \cite{Moreno2002, Wang2014}, but real-world contact networks evolve over time because of seasonal changes in human interaction patterns and in response to various situations (such as being sick, policies that ask people to ``shelter in place'' during a pandemic, and so on) \cite{Holme2012_temporal, Holme2015_temporal, Holme2019}. Such temporal changes in network structure can significantly impact the spread of a disease, and an important area of study is the dynamics of disease propagation on temporal networks \cite{Holme2016, Leitch2019, Masuda2013}. Leitch et al. \cite{Leitch2019} recently reviewed research that focuses on estimating epidemic thresholds on models of temporal networks. Existing approaches include neighbor-exchange models \cite{Volz2009}, activity-driven models \cite{Perra2012, Starnini2014}, and the use of a sequence of network snapshots \cite{aditya2010, valdano2015}. Different approaches can lead to the same formulation of an epidemic threshold. For example, Aditya et al. \cite{aditya2010} and Valdano et al. \cite{valdano2015} used different derivations to obtain the same formulation of the epidemic threshold of an SIS process on a temporal network. We discuss and compare their approaches in \Cref{sec:theoretical}. We also note that recent work has examined epidemic thresholds in certain continuous-time temporal networks. For instance, Valdano et al. \cite{valdano2018continuous} extended the approach in \cite{valdano2015} to a continuous-time setting in the special case in which adjacency matrices are ``weakly-commuting'' (specifically, when an adjacency matrix at a particular time commutes with an aggregated adjacency matrix up until that time).

Recently, Ahmad et al. \cite{ahmad2018tiedecay} introduced a type of temporal network that they called a \emph{tie-decay network}. Their approach, which draws on conceptual ideas from sociology \cite{Burt2000} and has some features in common with the model of social-network evolution in \cite{jin2001}, takes into account the fact that social relations experience continuous changes and decay in time. A tie-decay network distinguishes between ``ties'' and ``interactions'': a tie is a measurement of a social relationship between two entities that evolves continuously in time, whereas an interaction is some type of instantaneous contact between two entities. Unlike in most temporal network models, in which time has a discrete nature, a tie-decay network models ties between agents in a continuous manner. {A tie strengthens whenever there is an interaction between two entities, and it weakens between such interactions. For example, perhaps the strength of a tie decays exponentially following an interaction. This modeling assumption also is common in point-process models such as Hawkes processes \cite{laub2015,zipkin2016}.}

Because a tie-decay network is a type of temporal network with distinctive features, it is useful to investigate how standard dynamical systems, such as compartmental models of infectious diseases \cite{Brauer2019}, are affected by the structure of a tie-decay network. Studying a standard model (such as an SIS model of disease spread) on a tie-decay network allows one to examine how tie-decay networks affect dynamical processes that occur on them. In particular, by considering SIS dynamics on tie-decay networks, we seek to gain insights into models of disease spread on tie-decay networks. Many studies in network epidemiology examine the spread of diseases in a so-called ``quenched'' state (in which the spreading process is faster than the evolution of the network on which it spreads) or in a so-called ``annealed'' state (in which a network evolves more rapidly than a spreading process on it) \cite{porter2016dynamical}, but a tie-decay network need not possess such a separation into distinct time scales. On a tie-decay network, the evolution of the network and the spreading process can take place at comparable time scales. The tie strengths of a tie-decay network evolve continuously as a disease spreads, thereby influencing both the final outbreak size and the time at which the disease dies out or leads to a large-scale outbreak. Another way in which a tie-decay network differs from many other types of networks, such as those that arise from activity-driven models or when one just considers a sequence of network snapshots, is that tie strengths are not specified arbitrarily or determined by time-invariant activity rates that are associated with each node. Instead, the tie strengths in a network are governed both by the frequencies of the interactions between nodes and by the decay rates of these strengths. These features of tie-decay networks make them relevant for modeling social relationships, and we are thus motivated to investigate how these features influence the dynamics of disease spread.

In the study of the spread of infectious diseases, it is common to assume that a disease spreads only when two entities interact with each other  \cite{Pastor2015, kiss2017}. A tie-decay network is able to model the spread of a disease both {through direct} ``contact'' from close proximity (i.e., when an interaction takes place) and through indirect means (such as transmission through the air or by touching the same contaminated surface). A decaying tie can perhaps model the decrease in the likelihood of disease transmission following a direct interaction between individuals. For instance, when there is a direct contact between an infected individual and a susceptible one, a disease may not spread immediately from the former to the latter. It is also possible for disease transmission to occur after the susceptible individual touches an item that was exposed previously to the infected individual. Such indirect disease transmissions occur with lower probability as time elapses, and decaying ties between individuals can perhaps capture such situations. Employing tie-decay networks thus allows the possibility of disease spread even when there is no face-to-face interaction between entities. Studying disease dynamics on tie-decay networks can contribute to the understanding of how diseases spread in a real-world social network that evolves continuously in time. Additionally, compartmental models such as SIS processes have also been applied to studying the spread of information or attitudes in a population \cite{Guille2013,volk2020}. It seems potentially suitable to use tie-decay networks in such settings. For instance, an interaction can encode one entity informing another entity about some information, but the receiver of the information does not change their opinion until an ``infection'' event takes place. Additionally, because the strength of the tie between these two entities decays over time, the likelihood that the entity that receives information changes their beliefs from a new interaction between these two entities decreases with the amount of time since their last interaction. The use of tie-decay networks to study the spread of information or opinions allows one to differentiate an instantaneous transmission of information (through an interaction) from the long-lasting influence of information that was received in the past (as encoded in tie strength).

In the present paper, we study the dynamics of disease spread on a tie-decay network by examining the epidemic threshold of an SIS process. We first discuss the modeling choices that we make to associate the tie strengths with the spreading rates. Our mathematical formulation allows us to derive the epidemic threshold of an SIS process on a tie-decay network by extending the derivation of the epidemic threshold on other types of temporal networks. In our study, we build on methods that were designed for a sequence of network snapshots \cite{aditya2010, valdano2015}.\footnote{There has been some work on deriving epidemic thresholds on networks that evolve in continuous time (see, e.g.,  \cite{valdano2018continuous}), but their formulations make assumptions on network structure that do not apply to our tie-decay networks.} We then evaluate our theoretical expression for the epidemic threshold using numerical experiments on both synthetic and real-world networks, and we explore the impact of the network parameters on the spreading process. 

Our paper proceeds as follows. In \Cref{sec:model}, we mathematically formalize an SIS process on a tie-decay network. In \Cref{sec:theoretical}, we derive the epidemic threshold of an SIS process on a tie-decay network using two different methods. One is based on a nonlinear dynamical system, and the other is based on a tensor representation. In \Cref{sec:numerics}, we construct tie-decay networks from both synthetic and real-world data, and we simulate SIS processes on them. The results of the numerical experiments validate our theoretical expression for the epidemic threshold, and they also illustrate the influence of the network parameters on the spreading dynamics. In \Cref{sec:conclusion}, we conclude and propose future research directions.


\section{Model Setup}
\label{sec:model}

We first construct a tie-decay network using the formulation from \cite{ahmad2018tiedecay}. Let $\textbf{B}(t)$ be an $N \times N$ matrix of the tie strengths between each pair of the $N$ nodes in a network. The entry $b_{ij}(t)$ of $\textbf{B}(t)$ encodes the tie strength between nodes $i$ and $j$ at time $t$. Following an interaction between nodes $i$ and $j$, the strength of the tie between them decays exponentially according to the differential equation $b'_{ij}(t) = -\alpha b_{ij}(t)$, where $\alpha > 0$ is the decay coefficient. If nodes $i$ and $j$ interact at time $t$, then the strength of the tie between them increments by $1$ at time $t$. Therefore, if nodes $i$ and $j$ interact with each other at times $t = t_1, t_2, \ldots $, the tie strength between them satisfies
\begin{equation}
\label{eqn:continuous_evolve}
    b_{ij}(t) = b_{ij}(0)e^{-\alpha t} + \sum_{k: t_k < t}H(t-t_k){e^{-\alpha(t-t_k)}}\,,
\end{equation}
where $H(t)$ is the Heaviside step function. The following ordinary differential equation (ODE) describes the dynamics of the tie strengths:
\begin{equation}
    b'_{ij}(t) = -\alpha b_{ij} + \sum_{\{k: t_k < t\}}\delta(t-t_k){e^{-\alpha(t-t_k)}}\,.
\end{equation}
The interactions between nodes $i$ and $j$ are undirected in nature, so $b_{ij}(t) = b_{ji}(t)$ for all times $t$. {We do not consider self-interactions, so $b_{ii}(t) = 0$ for any node $i$ and any time $t$.} 

In practice, to model and analyze the spread of an infectious disease on a tie-decay network, we discretize time with a small time step of length $\Delta t$. Ahmad et al. \cite{ahmad2018tiedecay} chose a value of $\Delta t$ that is sufficiently small such that there is at most one interaction between agents. With this choice, one can convert a tie-decay network into a discrete set of temporal networks with adjacency matrices $\textbf{B}^{(\tau)} = \textbf{B}(\tau \Delta t)$, where $\tau = 0, 1, 2, \ldots$ indicates the time step. At each of these time steps, we suppose that the disease spreads, such that {any change} in network structure directly impacts the spreading properties at the $\tau$th time step. Although we discretize our tie-decay networks, we treat the underlying time as continuous. Additionally, we have the following relationship between a temporal snapshot and its predecessor:
\begin{equation}
	\textbf{B}^{(\tau)} = e^{-\alpha \Delta t}\textbf{B}^{(\tau-1)} + \textbf{A}^{(\tau)}\,,
\end{equation}
where $\textbf{A}^{(\tau)}$ is an indicator matrix in which either $2$ or $0$ entries are nonzero (because we are considering undirected networks). A pair of nonzero entries represents the single interaction that takes place during the $\tau$th time step. {Although Ahmad et al. \cite{ahmad2018tiedecay} required $\Delta t$ to be small enough such that there is at most one interaction in one time step, in practice, one can relax this requirement as long as the number of interactions in any time step $((\tau-1) \Delta t, \tau \Delta t]$ is not too large. 
In this way, we avoid binning interactions {into intervals of a fixed length. In \Cref{subsec:comparison_traditional}, we compare our results on a tie-decay network versus results that we obtain using a traditional temporal network in which we bin interactions into adjacent windows of a fixed length.} If there is an interaction between nodes $i$ and $j$ at time $t'$ that satisfies $(\tau-1) \Delta t < t' \leq \tau\Delta t$, we set $\textbf{A}^{(\tau)}_{ij} = \textbf{A}^{(\tau)}_{ji} = 1$ and set all other entries of $\textbf{A}^{(\tau)}$ to $0$. This interaction thus changes the network structure and influences the spreading behavior of a disease during the $\tau$th time step. In \Cref{subsec:real-world}, we discuss how we choose $\Delta t$ for networks that we construct from empirical data.} Because we discretize our tie-decay networks using a small time step, we obtain a number of temporal snapshots that tends to be much larger than the number of temporal snapshots that are often studied in practice in discrete-time temporal networks. In \Cref{subsec:choice_of_period}, we show that if we discretize a tie-decay network into $T$ temporal snapshots, it is possible to estimate the epidemic threshold using only the first $T_0 \ll T$ snapshots. This enables us to use a reasonable amount of computational time for studying disease dynamics on tie-decay networks.  

As a case study, we consider a susceptible--infected--susceptible (SIS) model \cite{kiss2017,Brauer2019} (one of the most common types of compartmental models), where the nodes can be either in a susceptible state or in an infected state (i.e., a ``compartment'' in the language of mathematical epidemiology). At each time step, a susceptible node can be infected by each of its infected neighbors with independent probability $\lambda$, and each infected node can recover from the disease and become susceptible again with independent probability $\mu$. We make a slight modification to the definition of a traditional SIS model to incorporate the traits of a tie-decay network. Suppose that an SIS process occurs on a tie-decay network with a tie-strength matrix $\textbf{B}(t)$ with entries $b_{ij}(t)$. We also assume that $\lambda_{\mathrm{max}}$ is the maximum possible infection probability and that the probability that an infected node $i$ infects a susceptible node $j$ at time $t$ is $\lambda_{\mathrm{max}}\mathrm{min}\{b_{ij}(t), 1\}$. That is, for nodes $i$ and $j$ with a tie strength $b_{ij}(t)$ that exceeds $1$, the infection probability is $\lambda_{\mathrm{max}}$. If the tie strength between them is {less than or equal to} $1$, then the infection probability is {$\lambda_{\mathrm{max}}b_{ij}(t)$}. When $b_{ij}(t) = 0$, there is no tie between nodes $i$ and $j$, so no infection event can take place between them. Therefore, the infection probabilities, which are different for different pairs of nodes, in a tie-decay network depend on how the network evolves in time. At each time $t$, an infected node recovers with probability $\mu$, and it is then in the susceptible state again at time $t+1$.

When modeling an SIS process on a tie-decay network, we first determine the duration of the time step $\Delta t$ in our discretization, and we discretize the network into a total of $T$ temporal snapshots. At each time step, we update the tie strengths $\textbf{B}^{(\tau)}$ by letting all ties decay exponentially and incrementing the {ties} for which an interaction takes place. For each infected node, we then see if there are any infection or recovery events. After the $T$th time step, the dynamics stop and we examine the final outbreak size of the epidemic. In \Cref{tbl:notations}, we summarize the main notation in our paper.

\begin{table}[htbp]
\begin{center}
\begin{tabular}{c | l p{8.5cm}}
Notation && Description \\\hline
$\textbf{B}(t)$ && the tie-strength matrix of a tie-decay network at time $t$ \\ \hline
$b_{ij}(t)$ && the (undirected) tie strength between entities (i.e., nodes) $i$ and $j$ at time $t$ \\ \hline
$\textbf{B}^{(\tau)}$ && $\textbf{B}^{(\tau)} = \textbf{B}(\tau \Delta t)$, the tie-strength matrix of a tie-decay network at the $\tau$th time step  after we discretize the network \\ \hline
$b_{ij}^{(\tau)}$ && the tie strength between entities $i$ and $j$ at the $\tau$th time step \\ \hline
$\textbf{A}^{(\tau)}$ && a symmetric matrix whose nonzero entries (there are either $0$ or $2$ of them) indicates what interaction takes place at the $\tau$th time step \\ \hline
$p_i^{(\tau)}$ && the probability that node $i$ is infected in the $\tau$th time step \\ \hline
$\lambda_{ij}^{(\tau)}$ && the infection probability between nodes $i$ and $j$ in the $\tau$th time step \\ \hline
$\lambda_{\mathrm{max}}$ && the maximum infection probability in the SIS process \\ \hline
${\Gamma(i)}$ && the set of neighbors of node $i$ \\ \hline
$\mu$ && the recovery probability in the SIS process \\ \hline
$N$ && the total number of nodes in a tie-decay network \\ \hline
{$p$} && the edge-creation probability of an Erd\H{o}s--R\'{e}nyi network \\ \hline
$\alpha$ && the decay coefficient of tie strengths in a tie-decay network \\ \hline
$\Delta t$ && the duration of one time step \\ \hline
$T$ && the total number of temporal snapshots after we discretize a tie-decay network \\ \hline
$l$ && the minimal length of the period for which the periodic boundary condition $\textbf{B}^{(\tau)} = \textbf{B}^{(\tau+l)}$ is satisfied \\ \hline
{$\textbf{S}$} && the system matrix of an SIS process on a tie-decay network \\ \hline
{$\rho_{\text{cr}}(\textbf{S})$} && the spectral radius of the matrix \textbf{S} (if $\textbf{S}$ is the system matrix of the SIS process, then $\rho_{\text{cr}}(\textbf{S})$ is the \emph{critical value} of the epidemic-threshold condition) \\ \hline
{$\mathcal{G}(N, p)$} && the ensemble of Erd\H{o}s--R\'{e}nyi networks with $N$ nodes and edge-creation probability $p$ \\ \hline
{$G(N, p)$} && an instance of an Erd\H{o}s--R\'{e}nyi network with $N$ nodes and edge-creation probability $p$ \\ \hline
\end{tabular}
\end{center}
\caption{The main notation in our paper.}
\label{tbl:notations} 
\end{table}


\section{Derivation of the Epidemic Threshold}
\label{sec:theoretical}

We now derive the epidemic threshold of an SIS process on a tie-decay network. The way that we perform time discretization allows us to extend methods that were designed for deriving epidemic thresholds in discrete temporal-network models to tie-decay networks. 

We derive the same epidemic threshold using two different methods. The first method is based on a derivation in Aditya et al.~\cite{aditya2010}, who modeled an SIS process as a nonlinear dynamical system and derived the epidemic threshold using linear stability analysis. The second method that we use was employed by Valdano et al.~\cite{valdano2015}, who modeled disease transmission using a multilayer representation \cite{kivela2014, DeDomenico2013} of a temporal network and an associated adjacency tensor. We discuss both methods to demonstrate two distinct approaches for deriving an epidemic threshold. We thereby illustrate that any method that one can apply to an arbitrary sequence of network snapshots is suitable for tie-decay networks because such methods utilize the temporal changes of a network. Both approaches rely on the use of a periodic boundary condition in time to derive an epidemic threshold. This is essential to guarantee stability of the disease-free state, and we thereby also use such boundary conditions in our derivations.


\subsection{Derivation using a Nonlinear Dynamical System}
\label{subsec:NLDS}

We derive the epidemic threshold of an SIS model on a tie-decay network by extending the approach in Aditya et al. \cite{aditya2010}, who modeled an SIS process using a nonlinear dynamical system. As we discussed in \Cref{sec:model}, we discretize a tie-decay network such that each time step is of length $\Delta t$. We thereby convert a tie-decay network into a discrete set of temporal networks with tie-strength matrices ${\textbf{B}^{(\tau)}}$, where $\tau \in \{1, 2, \ldots, T\}$. At each of the time steps, an infection of a susceptible node $j$ by an infected node $i$ occurs with probability $\lambda_{ij}^{(\tau)}=\lambda_{\mathrm{max}}\mathrm{min}\{b_{ij}^{(\tau)}, 1\}$ and each infected node recovers with independent probability $\mu$. Let $\xi_\tau(i)$ denote the probability that node $i$ does not become infected in the $\tau$th time step, and let $p_i^{(\tau)}$ denote the probability that node $i$ is in the infected state at time $\tau$. We use $\Gamma(i)$ to denote the set of neighbors of node $i$, and we assume that the states of the nodes in $\Gamma(i)$ are uncorrelated with each other. The following relationship holds:
\begin{equation}
\begin{aligned}
    \xi_\tau(i) 
    & = \prod_{j \in \Gamma(i)} \left(p_j^{(\tau)}(1-\lambda_{ij}^{(\tau)})+1-p_j^{(\tau)}\right)
    \\ & = \prod_{j \in \{1, \ldots, N \}} \left(1-\lambda_{ij}^{(\tau)}p_j^{(\tau)}\right)
    \\ & = \prod_{j \in \{1, \ldots, N \}} \left(1-\lambda_{\mathrm{max}}\mathrm{min}\{b_{ij}^{(\tau)}, 1\}p_j^{(\tau)}\right)\,.
\end{aligned}
\end{equation}
Additionally,
\begin{equation}
    1 - p_i^{(\tau+1)} = \mu p_i^{(\tau)} + (1-p_i^{(\tau)})\xi_\tau(i)\,,
\end{equation}
which implies that
\begin{equation}
    p_i^{(\tau+1)} = 1-\mu p_i^{(\tau)} - (1-p_i^{(\tau)})\prod_{j \in \{1, \ldots, N \} } \left(1-\lambda_{\mathrm{max}}\mathrm{min}\{b_{ij}^{(\tau)}, 1\}p_j^{(\tau)}\right)\,.
\label{eqn:relation}
\end{equation}
We let $\textbf{p}_\tau = (p_1^{(\tau)}, p_2^{(\tau)}, \ldots, p_N^{(\tau)})^\mathrm{T} $ and write \eqref{eqn:relation} in the form
\begin{equation}
    \textbf{p}_{\tau+1} = g_\tau(\textbf{p}_\tau)\,,
    \label{eqn:nlds}
\end{equation}
where $g_\tau$ is a function that depends on $\textbf{B}^{(\tau)}$. The nonlinear dynamical system \eqref{eqn:nlds} describes the dynamics of disease spread on a tie-decay network. To determine whether a disease dies out or leads to an outbreak, we assume that we have \emph{boundary periodic conditions} in time (see \cite{aditya2010}). That is, after we discretize a tie-decay network, we assume that $\textbf{B}^{(\tau)} = \textbf{B}^{(\tau+l)}$ for some constant $l$, {which allows us to examine the asymptotic stability of the system by looking at just one period}. Although a periodic boundary condition in time are not something that one expects to observe in temporal networks that one constructs from empirical data, we choose $l$ to be arbitrarily large so that we do not need to lose generality, because we can set $l = T$ (where we recall that $T$ is the final time). Additionally, in \Cref{subsec:choice_of_period}, we demonstrate that even with this periodic boundary condition, we can approximately characterize SIS spreading dynamics on a tie-decay network using a shorter period $l$. 

Given the discrete dynamical system \eqref{eqn:nlds}, we recall the following theorem \cite{hirsch2013}.
\begin{theorem}
The system $\textbf{x}_{\tau+1} = g(\textbf{x}_{\tau})$ is {asymptotically }stable at the fixed point $\textbf{x}^\ast$ if the magnitude of the dominant eigenvalue of the Jacobian $J = \nabla g(\textbf{x}^\ast)$ is less than $1$.
\label{thm:stability}
\end{theorem} 

Using this result gives the following lemma.
\begin{lemma}
Let $\textbf{S} = \prod_{\tau=0}^{l-1}{\textbf{S}_\tau}$, where $\textbf{S}_\tau = (1-\mu)\textbf{I} + \lambda_\mathrm{max}\mathrm{min}\{\textbf{B}^{(\tau)},1\}$. If the magnitude of the dominant eigenvalue of $\textbf{S}$ is less than $1$, then $\textbf{p}_\tau$ is asymptotically stable at $0$.
\label{lemma:nlds}
\end{lemma}

\begin{proof}
Because of the periodic boundary condition in time, it follows that $\textbf{S}_\tau = \textbf{S}_{\tau+l}$. The choice of $\tau$ is arbitrary, so it suffices to show that $\textbf{p}_{l\tau}$ is asymptotically stable at time $0$. Consider $\textbf{p}_{l(
\tau+1)} = g_{l-1}(g_{l-2}(\cdots(g_1(g_0(\textbf{p}_{l\tau})))\cdots))$. We have
\begin{equation}
\begin{aligned}
    \frac{\partial \textbf{p}_{l(\tau+1)}}{\partial \textbf{p}_{l\tau}} 
    & =  \left.\left(\frac{\partial \textbf{p}_{l(\tau+1)}}{\partial \textbf{p}_{l\tau+l-1}}\times \cdots \times \frac{\partial \textbf{p}_{l\tau+1}}{\partial \textbf{p}_{l\tau}}\right)\right|_{\textbf{p}_{l\tau}=0}
    \\ & = \left(\left.\frac{\partial \textbf{p}_{l(\tau+1)}}{\partial \textbf{p}_{l\tau+l-1}}\right|_{\textbf{p}_{l\tau+l-1}=0}\right) \times \cdots \times \left(\left.\frac{\partial \textbf{p}_{l\tau+1}}{\partial \textbf{p}_{l\tau}}\right|_{\textbf{p}_{l\tau}=0}\right)
    \\ & = \textbf{S}_{l-1} \times \cdots \times \textbf{S}_0
    \\ & = \textbf{S}\,.
\end{aligned}
\end{equation}
Consequently, by \Cref{thm:stability}, if the magnitude of the dominant eigenvalue of $\textbf{S}$ is less than $1$, it follows that $\textbf{p}_{l\tau}$ is asymptotically stable at $0$. We also obtain asymptotic stability of $\textbf{p}_{l\tau+1}, \ldots, \textbf{p}_{l\tau+l-1}$ because the dominant eigenvalue of the product of invertible matrices is invariant under a cyclic permutation.
\end{proof}

\medskip

If the magnitude of the dominant eigenvalue of each matrix $\textbf{S}_\tau$ is less than $1$, then the magnitude of the dominant eigenvalue of $\textbf{S}$ is also less than $1$ and the system is asymptotically stable at $0$. However, this is a much more conservative condition than our epidemic-threshold condition:
\begin{equation}
	\rho_{\text{cr}}(\textbf{S}) = \rho_{\text{cr}}\left(\prod_{\tau=0}^{l-1}{\textbf{S}_\tau}\right)=1\,,
\label{eqn:epidemic_threshold}
\end{equation}
where $\rho_{\text{cr}}(\mathbf{\Theta})$ is the spectral radius of the matrix $\mathbf{\Theta}$. Even when the dominant eigenvalues of some of the $\textbf{S}_\tau$ matrices have magnitudes that are larger than $1$, the disease can still die out asymptotically, depending on the spectrum of $\textbf{S}$. We refer to $\textbf{S}$ as the \emph{system matrix} of our SIS process on a tie-decay network, and we refer to $\rho_{\text{cr}}(\textbf{S})$ as the \emph{critical value} of the epidemic-threshold condition (and hence of the epidemic threshold).


\subsection{Derivation using a Multilayer Representation}

We now derive the epidemic threshold of an SIS process on a tie-decay network by extending the derivation in \cite{valdano2015} that is based on a multilayer representation of a temporal network. We again consider our SIS model on a tie-decay network with tie-strength matrix $\textbf{B}(t)$, which we discretize into snapshots ${\textbf{B}^{(\tau)}}$, where $\tau \in \{1, 2, \ldots, T\}$. As in \Cref{subsec:NLDS}, we use the periodic boundary condition $\textbf{B}^{(\tau)} = \textbf{B}^{(\tau+l)}$ with period $l$ and we seek the asymptotic solution for one period. We define a tensor $\widetilde{\textbf{M}}$ with components 
\begin{equation}
    \widetilde{\textbf{M}}^{\tau \tau'}_{ij} = \delta_{\tau,\tau'-1}[(1-\mu)\delta_{ij}+\lambda_\mathrm{max} \mathrm{min}\{b_{ij}^{(\tau)},1\}]
\end{equation}
to reflect the dynamics of \eqref{eqn:relation}. This tensor encodes information about the probability that node $i$ is infected by node $j$ when we advance from time step $\tau$ to time step $\tau+1$. One can also represent $\widetilde{\textbf{M}}$ using a supra-adjacency matrix $\textbf{M} \in \mathbb{R}^{Nl \times Nl}$ (see \cite{DeDomenico2013, Diaz2013}), where $N$ is the total number of nodes. Let $\textbf{p}(k)$ be the state vector of the $k$th time period, which covers time steps in the interval $[kl, (k+1)l]$. Let $\alpha = N\tau + i$, and let $\textbf{p}_\alpha(k)$ denote the probability that node $i$ is in the infected state in time step $kl+\tau$. Using this notation, we write \eqref{eqn:relation} as
\begin{equation}
    \textbf{p}_\alpha(k) = 1-\prod_\beta[1-\textbf{M}_{\beta\alpha}\textbf{p}_\beta(k-1)]\,.
\end{equation}
The asymptotic solution of {the state vector $\hat{\textbf{p}}$ for} one period  is 
\begin{equation}
    \hat{\textbf{p}}_\alpha = 1-\prod_\beta[1-\textbf{M}_{\beta\alpha}\hat{\textbf{p}}_\beta]\,.
\end{equation}
By the analysis in \cite{valdano2015}, the epidemic-threshold condition is
\begin{equation}
    \rho_{\text{cr}}(\textbf{M}) = \rho_{\text{cr}}(\textbf{S})^{1/l} = 1\,,
\label{eqn:critical_value}
\end{equation}
where $\textbf{S} = \prod_{\tau=1}^l[(1-\mu)\textbf{I}+\lambda_\mathrm{max} \mathrm{min}\{\textbf{B}^{(l-\tau)},1\}]$, which matches the formulation in \Cref{lemma:nlds}. This again yields the epidemic-threshold condition \eqref{eqn:epidemic_threshold}.


\section{Numerical Experiments}
\label{sec:numerics}

{We now conduct numerical experiments} in which we simulate an SIS process on various tie-decay networks. We perform our computations on a workstation using code that we wrote in {\sc Python}\footnote{Our code is publicly available on GitHub at {\tt https://github.com/qinyichen/tie-decay-epidemic-threshold}.}. In \Cref{subsec:validation}, we validate the epidemic-threshold condition \eqref{eqn:epidemic_threshold} by comparing our theoretical results with our numerical simulations. In \Cref{subsec:influence_factors}, we construct tie-decay networks with different decay coefficients, temporal interactions, and sparsities\footnote{{We use Erd\H{o}s--R\'{e}nyi networks in our experiments. We say that an Erd\H{o}s--R\'{e}nyi network is ``sparser'' than another Erd\H{o}s--R\'{e}nyi network if the edge-creation probability $p$ is smaller in that network than in the other network.}} to explore how these factors influence the outcome of disease spread. In \Cref{subsec:choice_of_period}, we discuss the periodic boundary condition $\textbf{B}^{(\tau)} = \textbf{B}^{(\tau+l)}$ that we stated in \Cref{sec:theoretical} and check numerically how the choice of period $l$ affects the epidemic threshold. In \Cref{subsec:real-world}, we explore disease dynamics on tie-decay networks that we construct from empirical data. Finally, in \Cref{subsec:comparison_traditional}, we compare the results on a tie-decay network that we construct from an Erd\H{o}s--R\'{e}nyi (ER) network with those on a traditional temporal network that we construct from binning interactions of the same ER network. This further motivates the use of tie-decay networks as a viable modeling framework for studying dynamical processes on temporal networks. 

To simplify our notation, we use $\lambda$ to denote the maximum infection probability $\lambda_\mathrm{max}$ throughout this section.


\subsection{Validation of Our Epidemic Threshold}
\label{subsec:validation}

\begin{figure}[htbp]
\centering
\subfloat[Outbreak sizes at the end of our simulations.]{%
  \includegraphics[width=0.8\textwidth]{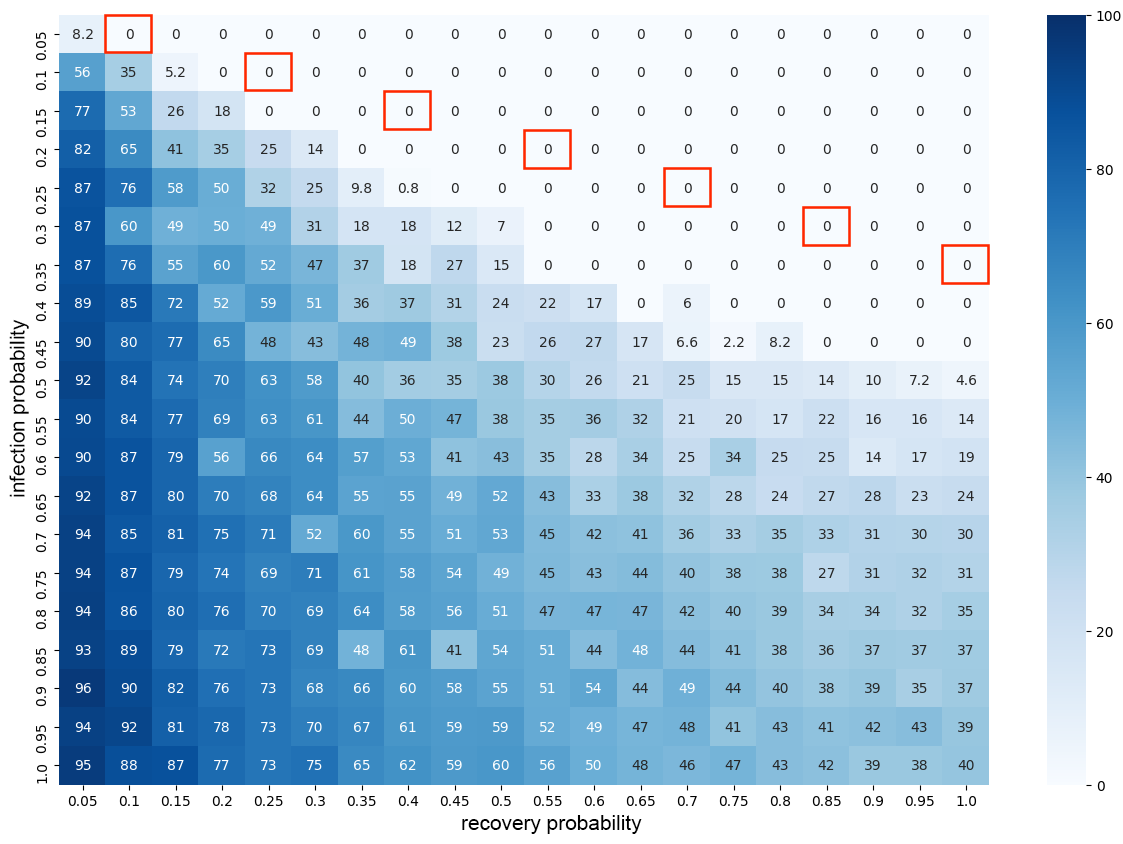}%
  \label{fig:validation-outbreak-size}
}
\qquad
\subfloat[Critical values $\rho_{\text{cr}}(\textbf{S})$ of the epidemic-threshold condition.]{%
  \includegraphics[width=0.8\textwidth]{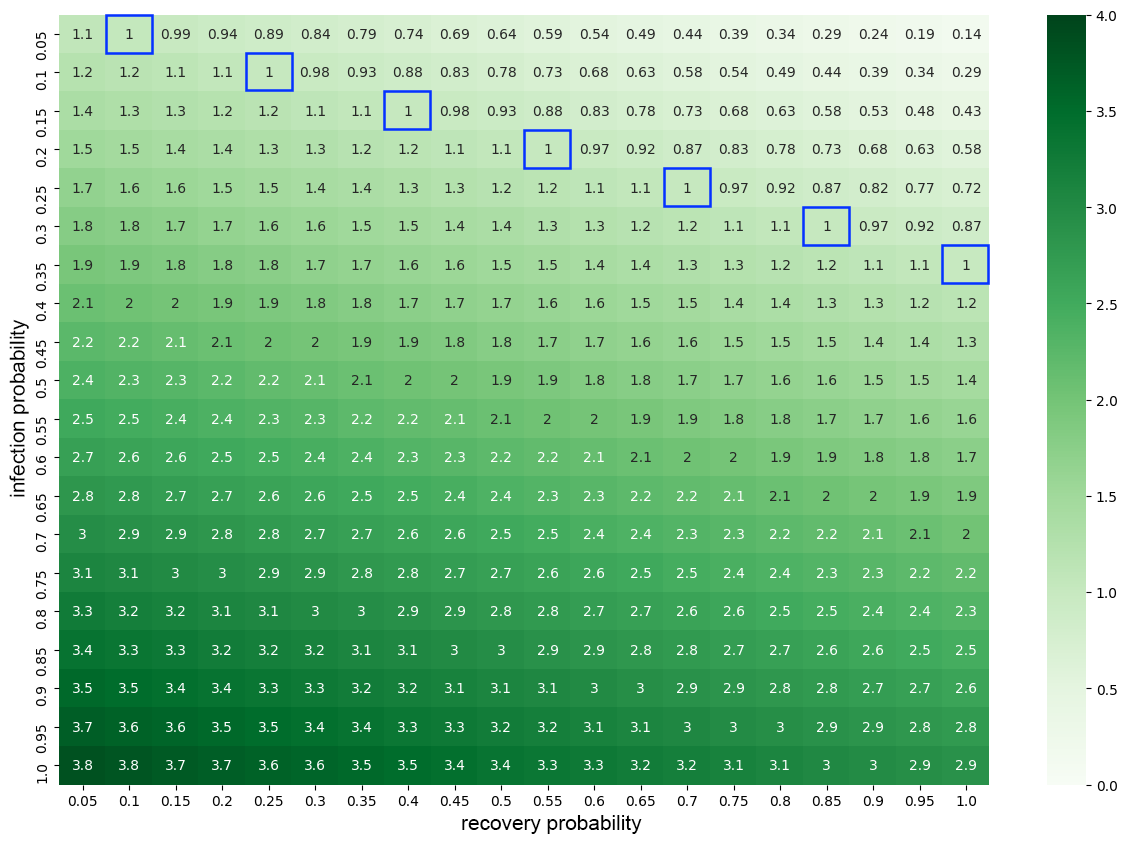}%
  \label{fig:validation-critical-value}
}
\caption{(a) The final outbreak sizes in our numerical simulations of an SIS model on a tie-decay network that we construct from an Erd\H{o}s--R\'{e}nyi (ER) network and (b) the associated critical values of the epidemic-threshold condition. (See the text for more details.) For each numerical simulation, we simulate an SIS process using each pair of infection and recovery probabilities for $10^3$ time steps. We do each simulation $10$ times, and we report the means of the final outbreak sizes.}
\label{fig:validation}
\end{figure}

To validate the epidemic-threshold condition \eqref{eqn:epidemic_threshold} in \Cref{sec:theoretical}, we construct a tie-decay network from an Erd\H{o}s--R\'{e}nyi (ER) network in the following way. {Let $G(N, p)$ be an instance of the $\mathcal{G}(N, p)$ ensemble of ER networks}, where $N$ is the number of nodes and $p$ is the probability of an edge between each pair of nodes. To create a tie-decay network, we assign a sequence $T_e = t_1, t_2, \ldots$ of time stamps to each edge $e = (i, j)$; these time stamps indicate the times of the interactions between nodes $i$ and $j$. We generate the sequences of time stamps using an exponential waiting-time distribution with scale $\beta$. That is, the difference $t_{k+1} - t_{k}$ between two consecutive event times, $t_k$ and $t_{k+1}$, has a mean of $\beta$. We suppose that our tie-decay network starts from a tie-strength matrix $\textbf{B}(0)$ with all edges of equal tie strength $0.5$. The tie strengths then evolve continuously following \eqref{eqn:continuous_evolve}. We increment the tie strength of edge $e$ at each $t_{i} \in T_e$, and the strength of a tie decays exponentially when there are no interactions. We then discretize the tie-decay network (see our discussion in \Cref{sec:model}) and simulate an SIS process with {maximum infection probability} $\lambda$ and recovery probability $\mu$ on this network. 

In our numerical experiments in \Cref{fig:validation}, we construct a tie-decay network from an instance $G^{(1)}_{\text{ER}} = G(100, 0.1)$ of an ER network {(where we ensure that $G^{(1)}_{\text{ER}}$ has a single connected component)} and a tie-decay coefficient of $\alpha = 10^{-1}$.  We generate interactions with an exponential distribution with scale $\beta = 100$, and the simulations each last $10^3$ time steps. We simulate our SIS process on this tie-decay network with various infection and recovery probabilities that each range from $0.05$ to $1$. In \Cref{fig:validation}, we show the outbreak sizes that we obtain at the end of our simulations and the associated critical value $\rho_{\text{cr}}(\textbf{S})$ of the epidemic-threshold condition \eqref{eqn:epidemic_threshold}. In our examination of different pairs of infection and recovery probabilities, we observe transitions in both the final outbreak sizes and the critical values. In \Cref{fig:validation-critical-value}, for a fixed maximum infection probability $\lambda$, we highlight the recovery probability $\mu$ for which the critical value is closest to $1$. We then highlight the same $(\lambda, \mu)$ pairs in \Cref{fig:validation-outbreak-size}. We observe for all of the $(\lambda, \mu)$ pairs that yield critical values less than $1$ that the disease always dies out by the end of our simulations. This supports our theoretical result \eqref{eqn:epidemic_threshold} that if $\rho_{\text{cr}}(\textbf{S})$ is less than $1$, then the disease-free state (in which no nodes are infected) is asymptotically stable. Because we perform our simulations on a network with finitely many nodes over finitely many time steps, there are some $(\lambda, \mu)$ pairs for which the critical values are slightly larger than $1$ but have $0$ final outbreak sizes. Typically, however, we observe that after critical values exceed $1$, a larger critical value usually corresponds to a larger final outbreak size at the end of our simulations. 

To further illustrate the correlation between the final outbreak size and the critical value of the epidemic-threshold condition, we plot their relationships as a scatter plot. In \Cref{fig:scatter-plot}, we again work with the tie-decay network that we constructed from the network $G^{(1)}_{\text{ER}}$, the decay coefficient $\alpha = 10^{-1}$, and the scale $\beta = 100$. In the scatter plot, we see that when we reach disease-free state (i.e., when the final outbreak size is $0$), most of the critical values of the epidemic-threshold condition are less than $1$, which again agrees with our theoretical results in \eqref{eqn:epidemic_threshold}. Additionally, when the final outbreak size exceeds $0$, its value appears to be positively correlated with the critical value $\rho_{\text{cr}}(\textbf{S})$ of the epidemic-threshold condition \eqref{eqn:epidemic_threshold}. 
We also study this correlation on tie-decay networks that we construct from the same network $G^{(1)}_{\text{ER}}$ with different decay coefficients ($\alpha = 10^{-1}$, $\alpha = 10^{-2}$, and $\alpha = 10^{-3}$) and scales ($\beta = 10$, $\beta = 50$, and $\beta = 100$). 
In \Cref{tbl:correlation}, we show the Pearson correlation coefficient (PCC) between the final outbreak size and the critical value for each of these tie-decay networks. All of our scenarios have a PCC of at least $0.5$, which confirms the strong positive correlation between the final outbreak size and the critical value $\rho_{\text{cr}}(\textbf{S})$.

\begin{figure}[htbp]
  \centering
  \subfloat[A scatter plot of the final outbreak size versus the critical value $\rho_{\text{cr}}(\textbf{S})$ for the tie-decay network that we construct from the ER network $G^{(1)}_{\text{ER}}$ with decay coefficient $\alpha = 10^{-1}$ and a waiting-time distribution with mean $\beta = 100$. We highlight the epidemic-threshold condition $\rho_{\text{cr}}(\textbf{S}) = 1$ using a dashed red line.]{%
  \includegraphics[width=0.45\textwidth,valign=c]{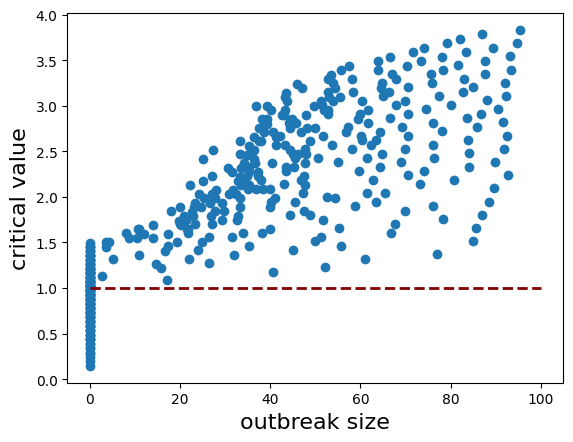}%
  \label{fig:scatter-plot}
  }
  \qquad
  \subfloat[The Pearson correlation coefficient (PCC) between the final outbreak size and the critical value. We compute the sample PCC for tie-decay networks that we construct from the ER network $G^{(1)}_{\text{ER}}$ with different values of $\alpha$ and $\beta$.]{
  \adjustbox{valign=c}
  {
  \begin{tabular}{ llr } \toprule
    $\alpha$  & $\beta$  & PCC \\ \midrule
    \multirow{3}{*}[-1ex]{$10^{-1}$} 
    & 10 &  0.606 \\
    & 50 &  0.768 \\
    & 100 & 0.825 \\
    \cmidrule(lr){1-3}
    \multirow{3}{*}[-1ex]{$10^{-2}$} 
    & 10 & 0.564 \\
    & 50 & 0.612 \\
    & 100 & 0.646 \\
    \cmidrule(lr){1-3}
    \multirow{3}{*}[-1ex]{$10^{-3}$} 
    & 10 & 0.590 \\
    & 50 & 0.601 \\
    & 100 & 0.637 \\
    \bottomrule
    \end{tabular}
    \label{tbl:correlation}
  }}
  \caption{Illustrations of the correlation between the final outbreak size and the critical value $\rho_{\text{cr}}(\textbf{S})$.
  }%
  \label{fig:table}%
\end{figure}


\subsection{Influence of Tie-Decay Networks and Their Parameters on Disease Spread}
\label{subsec:influence_factors}

We just demonstrated (see \Cref{subsec:validation}) that the final outbreak size and the critical value $\rho_{\text{cr}}(\textbf{S})$ of the epidemic-threshold condition \eqref{eqn:epidemic_threshold} has a strong, positive correlation. Therefore, we can potentially use this critical value can potentially as an indicator of the scale of the spread of a disease. Because the critical value is an important quantity, we study how different factors influence disease spread on tie-decay networks by comparing their critical values. There are three primary parameter choices that influence the spreading dynamics: (1) the tie-decay coefficient $\alpha$, which determines how fast tie strengths decay; (2) the interaction frequency between nodes, which one can tune using the scale $\beta$ of the exponential waiting-time distribution; and (3) the sparsity of the underling network, which we determine using the edge-creation probability $p$ of an ER network. We have an intuitive expectation of how each of these features influences disease dynamics. For instance, when interactions take place more frequently, one usually expects a disease to spread more easily and hence to infect more people. When a network is sparse (i.e., there are many fewer edges in it than the maximum possible number of edges), it tends to be more difficult for a disease outbreak to occur. By computing the critical values of SIS processes on different tie-decay networks, we examine if our intuition is correct.

\paragraph{Decay Coefficient.} We construct tie-decay networks using one network $G^{(2)}_{\text{ER}}$ of the $\mathcal{G}(N, p)$ ER network ensemble with $N = 100$ nodes and edge-creation probability $p = 0.05$ {(where we ensure that $G^{(2)}_{\text{ER}}$ has a single connected component)}, and we generate time stamps for each edge using an exponential waiting-time distribution with scale $\beta = 100$. We then create three variants of this tie-decay network by using decay coefficients of $\alpha = 10^{-1}$, $\alpha = 10^{-2}$, and $\alpha = 10^{-3}$. In \Cref{fig:critical-val-decay-coeff}, we compute the critical values of SIS processes with different infection rates and recovery probabilities for each of these tie-decay networks. As in \Cref{subsec:validation}, we highlight the $(\lambda, \mu)$ pairs that have critical values that are closest to $1$. This enables us to roughly divide the $(\lambda, \mu)$ parameter plane into two regions. For $(\lambda, \mu)$ pairs in the upper-right part of each plot in \Cref{fig:critical-val-decay-coeff}, the disease eventually dies out. For the rest of the $(\lambda, \mu)$ pairs, the initial infection tends to result in an outbreak. In \Cref{fig:critical-val-decay-coeff}, we observe that there are many more $(\lambda, \mu)$ pairs for which the disease dies out for $\alpha = 10^{-1}$ than for $\alpha = 10^{-2}$ and $\alpha = 10^{-3}$. For progressively smaller values of $\alpha$, it becomes more likely for an outbreak to occur. A larger decay coefficient $\alpha$ leads to stronger tie strengths in the long run. (See the discussion in \cite{ahmad2018tiedecay}.) This, in turn, makes it easier for a disease to spread because the transmission of an infection between two nodes is positively correlated with the strength of the tie between them.

\begin{figure}[htbp]
\centering
\includegraphics[width=0.8\textwidth]{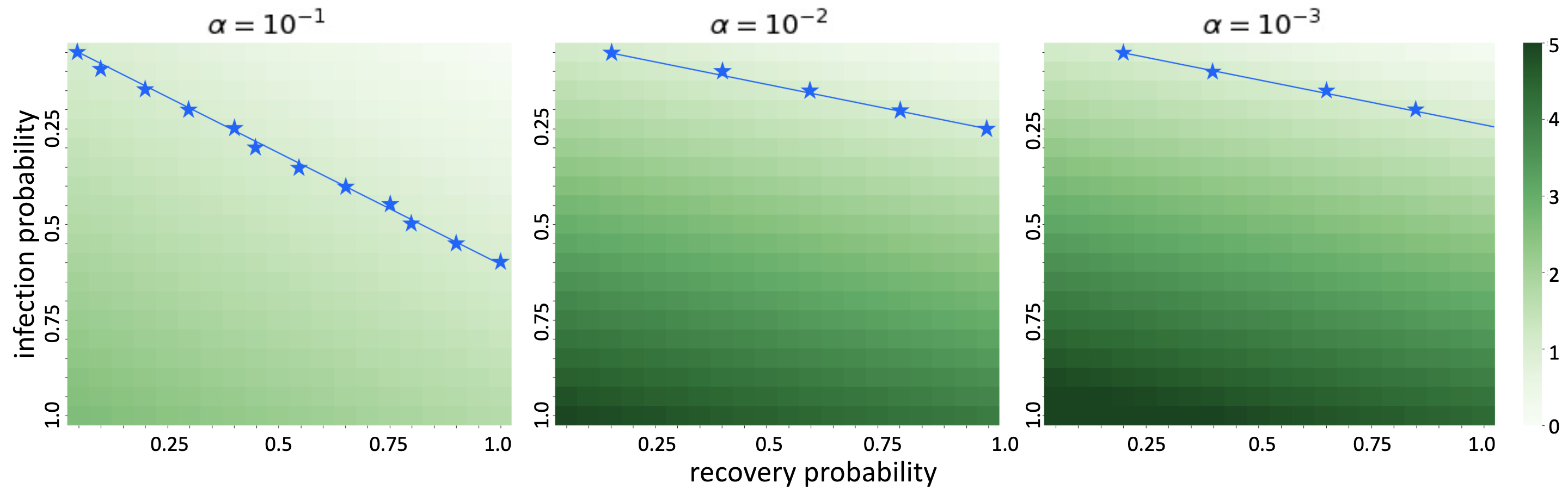}
\caption{The critical values $\rho_{\text{cr}}(\textbf{S})$ (with darker colors signifying larger values) that we compute for tie-decay networks with decay coefficients of (left) $\alpha = 10^{-1}$, (center) $\alpha = 10^{-2}$, and (right) $\alpha = 10^{-3}$ that we construct from the ER network $G^{(2)}_{\text{ER}}$. For each fixed value of the maximum infection probability $\lambda$, the star symbol indicates the smallest recovery probability $\mu$ that gives a critical value that is closest to $1$. The line signifies the rough boundary between critical values that are larger than $1$ and those that are smaller than $1$. We draw the lines manually to guide human eyes; we do not generate them using either mathematical reasoning or computations.}
\label{fig:critical-val-decay-coeff}
\end{figure}

\paragraph{Interaction frequency.} We also examine the influence of interaction frequency on disease spread in our SIS model. We construct tie-decay networks using the same ER network $G^{(2)}_{\text{ER}}$ and a decay coefficient of $\alpha = 10^{-2}$. For each edge, we generate time stamps with an exponential waiting-time distribution with scales $\beta = 10$, $\beta = 50$, and $\beta = 100$. The interactions between nodes is the most frequent when $\beta = 10$; in this case, the mean time between two consecutive interactions is $10 \Delta t$, where $\Delta t$ is the duration of a time step. In \Cref{fig:critical-val-temporal}, we observe that the dividing line for the epidemic threshold shifts gradually to the left for progressively larger values of $\beta$. We also observe this in the colors of the heat maps, for which a darker green indicates a larger critical value. For progressively larger values of $\beta$ (i.e., for decreasingly frequent interactions between nodes), the number of grids that are covered in dark green also becomes smaller. In other words, when interactions between nodes occur more frequently, it is easier for a disease to spread through a population.

\begin{figure}[htbp]
\centering
\includegraphics[width=0.8\textwidth]{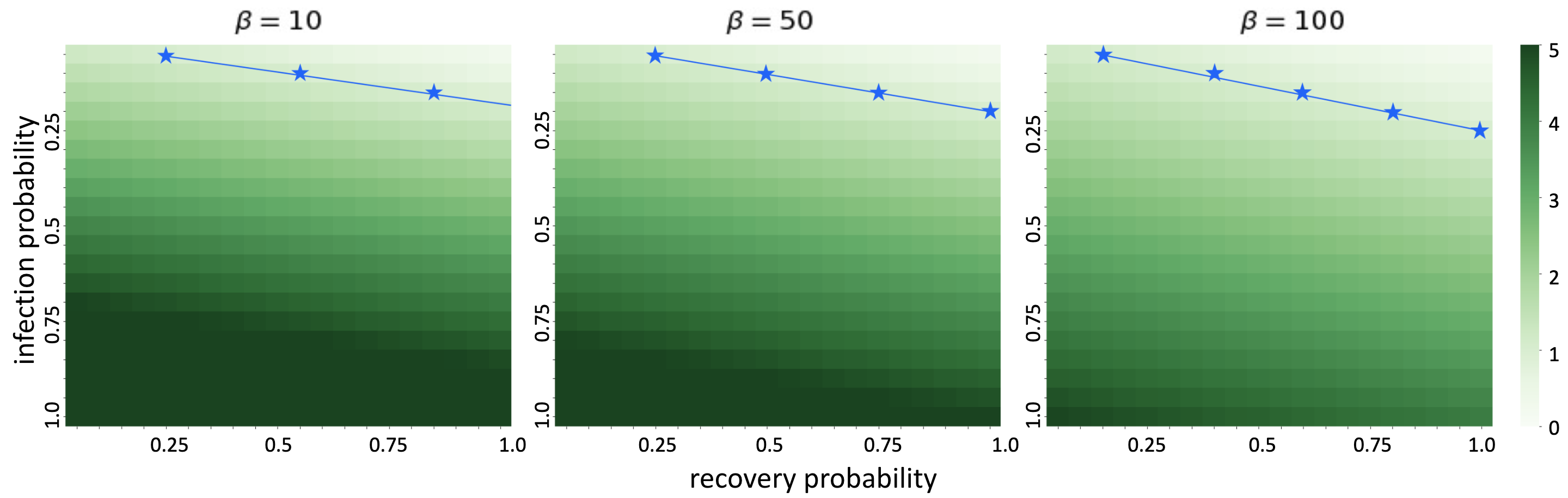}
\caption{The critical values $\rho_{\text{cr}}(\textbf{S})$ (with darker colors signifying larger values) for tie-decay networks with different interaction frequencies. We construct the tie-decay networks from the ER network $G^{(2)}_{\text{ER}}$ and generate interactions from exponential waiting-time distributions with scales (left) $\beta = 10$, (center) $\beta = 50$, and (right) $\beta = 100$. The stars and lines have the same meaning as in \Cref{fig:critical-val-decay-coeff}.}
\label{fig:critical-val-temporal}
\end{figure}

\paragraph{Sparsity of the Networks.} We construct tie-decay networks using three networks from the $\mathcal{G}(N, p)$ ER network ensemble with $N = 100$. The network $G^{(1)}_{\text{ER}}$ (which we examined previously) has an edge-creation probability of $p = 0.10$, the network $G^{(2)}_{\text{ER}}$ (which we also examined previously) has an edge-creation probability of $p = 0.05$, and the network $G^{(3)}_{\text{ER}}$ has an edge-creation probability of $p = 0.02$. We ensure that each of the three networks consists of a single connected component. We generate the time stamps for each edge from exponential distributions with decay coefficient $\alpha = 10^{-2}$ and scale $\beta = 100$. We compare the dividing line of the epidemic threshold and the colors of the heat maps in \Cref{fig:critical-val-network-sparsity}. In the densest tie-decay network (which we construct using $G^{(1)}_{\text{ER}}$), almost all $(\lambda, \mu)$ pairs lead to an eventual outbreak of the disease. By contrast, for sparser tie-decay networks, such as the one that we construct from $G^{(3)}_{\text{ER}}$, outbreaks are less likely to occur. This matches our intuition about SIS disease dynamics on tie-decay networks with different sparsities.

\begin{figure}[H]
\centering
\includegraphics[width=0.8\textwidth]{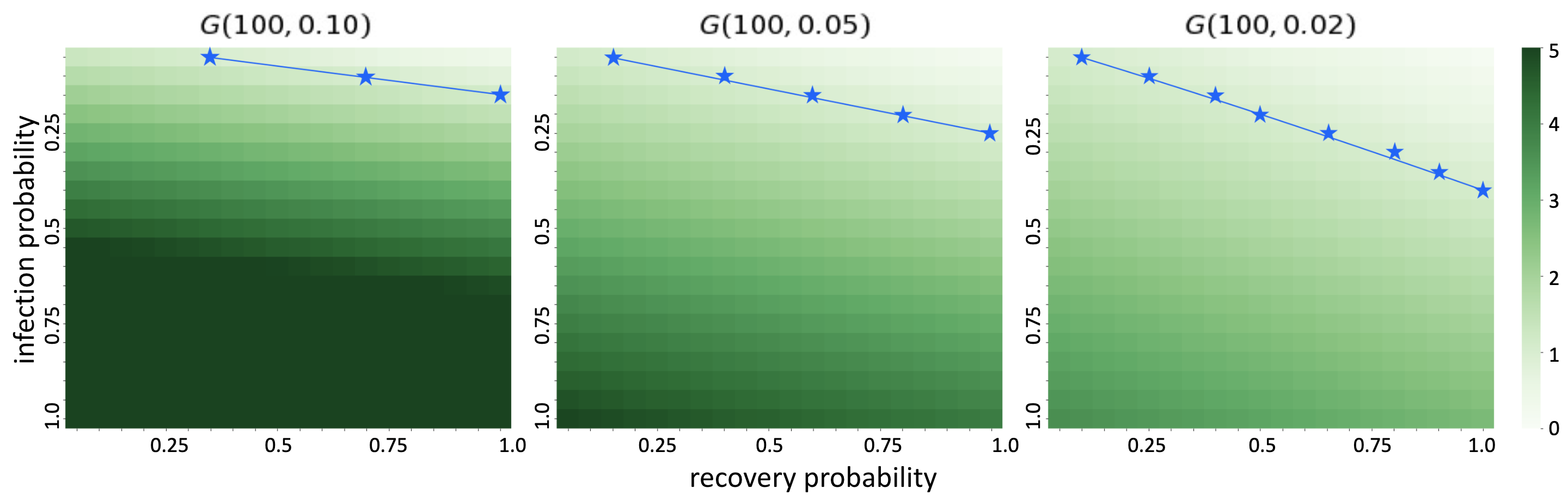}
\caption{The critical values $\rho_{\text{cr}}(\textbf{S})$ (with darker colors signifying larger values) that we compute for tie-decay networks that we construct from ER networks of different sparsities. Each of the three networks has $N = 100$ nodes; their edge-creation probabilities are $p = 0.10$, $p = 0.05$, and $p = 0.02$. The star symbols and the lines have the same meanings as in \Cref{fig:critical-val-decay-coeff}.}
\label{fig:critical-val-network-sparsity}
\end{figure}


\subsection{Choice of the Time Period for Examining the Epidemic Threshold}
\label{subsec:choice_of_period}

In \Cref{sec:theoretical}, we used the periodic boundary condition $\textbf{B}^{(\tau)} = \textbf{B}^{(\tau+l)}$, which requires the tie-strength matrix to be periodic in time with period $l$. However, for most tie-decay networks, such periodic behavior does not occur. Tie strengths increment instantaneously and decay continuously in time, so it would be very surprising for such periodicity to occur. Valdano et al. \cite{valdano2015} proposed that as long as the data-collection time period $l$ is long enough, the data gives `an approximately complete reconstruction of the temporal network properties', and once hence ought to be able to accurately estimate the epidemic threshold of a contagion model on a temporal network, even if it constructed from empirical data. In our tie-decay networks, we demonstrate using numerical computations that one can characterize the outcome of an entire SIS process by computing the epidemic threshold over a time period that is smaller than the entire time span.

We simulate two SIS processes on a tie-decay network that we construct from the ER network $G^{(2)}_{\text{ER}}$ with a decay coefficient of $\alpha = 10^{-1}$. We generate the interactions from an exponential waiting-time distribution with scale $\beta = 100$. The first SIS process has a {maximum} infection probability of $\lambda = 0.3$ and a recovery probability of $\mu = 0.7$, and the second SIS process has a {maximum} infection probability of $\lambda = 0.4$ and a recovery probability of $\mu = 0.6$. We simulate each SIS process for $10^3$ time steps. If we compute the critical threshold $\rho_{\text{cr}}(\textbf{S})$ using the period $l = 10^3$, we obtain $\rho_{\text{cr}}(\textbf{S}_\textbf{1}) \approx 0.816$ for the first SIS process and $\rho_{\text{cr}}(\textbf{S}_\textbf{2}) \approx 1.088$ for the second SIS process. In \Cref{fig:choice-of-period}, we plot the evolution of the critical values for different choices of the time period $l$. In \Cref{fig:choice-of-period-1}, we see that although the critical value starts above the threshold and changes rapidly at first, it stabilizes after a fairly small number of time steps. When $(\lambda,\mu) = (0.3,0.7)$, we observe for all $l \geq 140$ that all of the critical values lie in the interval $(0.80, 0.82)$. In \Cref{fig:choice-of-period-2}, the critical values again converge quickly after a small number of time steps. When $(\lambda,\mu) = (0.4,0.6)$, we observe for all $l \geq 143$ that all of the critical values lie in the interval $(1.07, 1.09)$. From these two examples, we see that regardless of whether the critical value $\rho_{\text{cr}}(\textbf{S})$ is above $1$ or below $1$, we are able to accurately estimate the epidemic threshold by using a period $l$ that is fairly small in comparison to the length of the entire time span. The fast convergence of the critical values is a feature of the employed tie-decay network model. Valdano et al.~\cite{valdano2015} studied the influence of the time period on estimations of the epidemic threshold of an SIS process on a multilayer representation of a temporal network for a sequence of temporal snapshots. However, in their experiments, when the time period $l$ is small in comparison to the length of the total time span, they did not always observe convergence of the critical values. The fast convergence of the critical values on our tie-decay networks enables us to estimate the outcome of an SIS process using data from only the early stages of an epidemic. Specifically, by calculating the epidemic threshold using \eqref{eqn:epidemic_threshold}, one can potentially characterize the spreading dynamics of an epidemic that lasts for several years using 
the data from the first hundred days.

\begin{figure}[htbp]
\centering
\subfloat[$\lambda = 0.3$, $\mu = 0.7$]{%
  \includegraphics[width=0.45\textwidth]{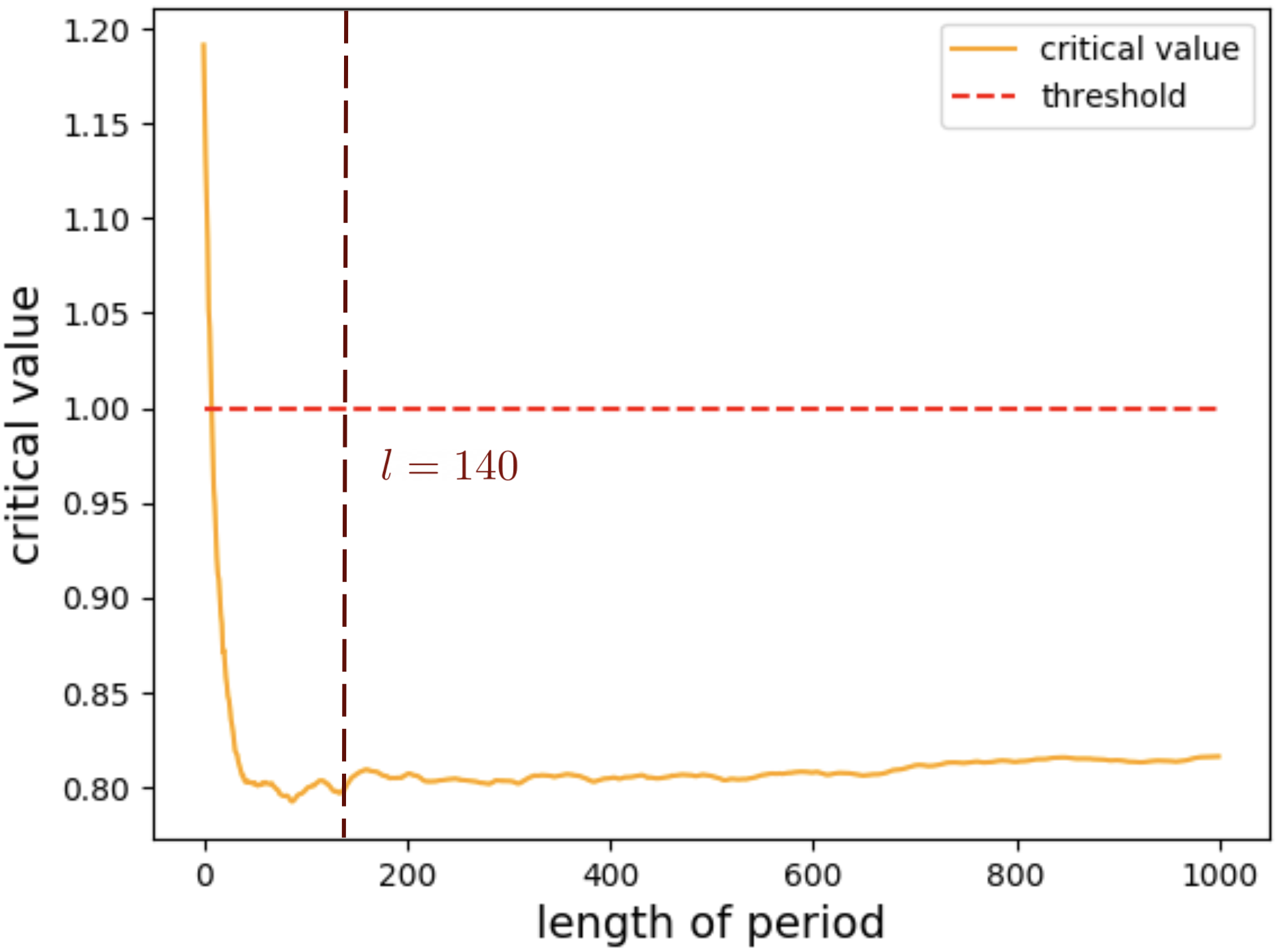}%
  \label{fig:choice-of-period-1}
}\qquad
\subfloat[$\lambda = 0.4$, $\mu = 0.6$]{%
  \includegraphics[width=0.45\textwidth]{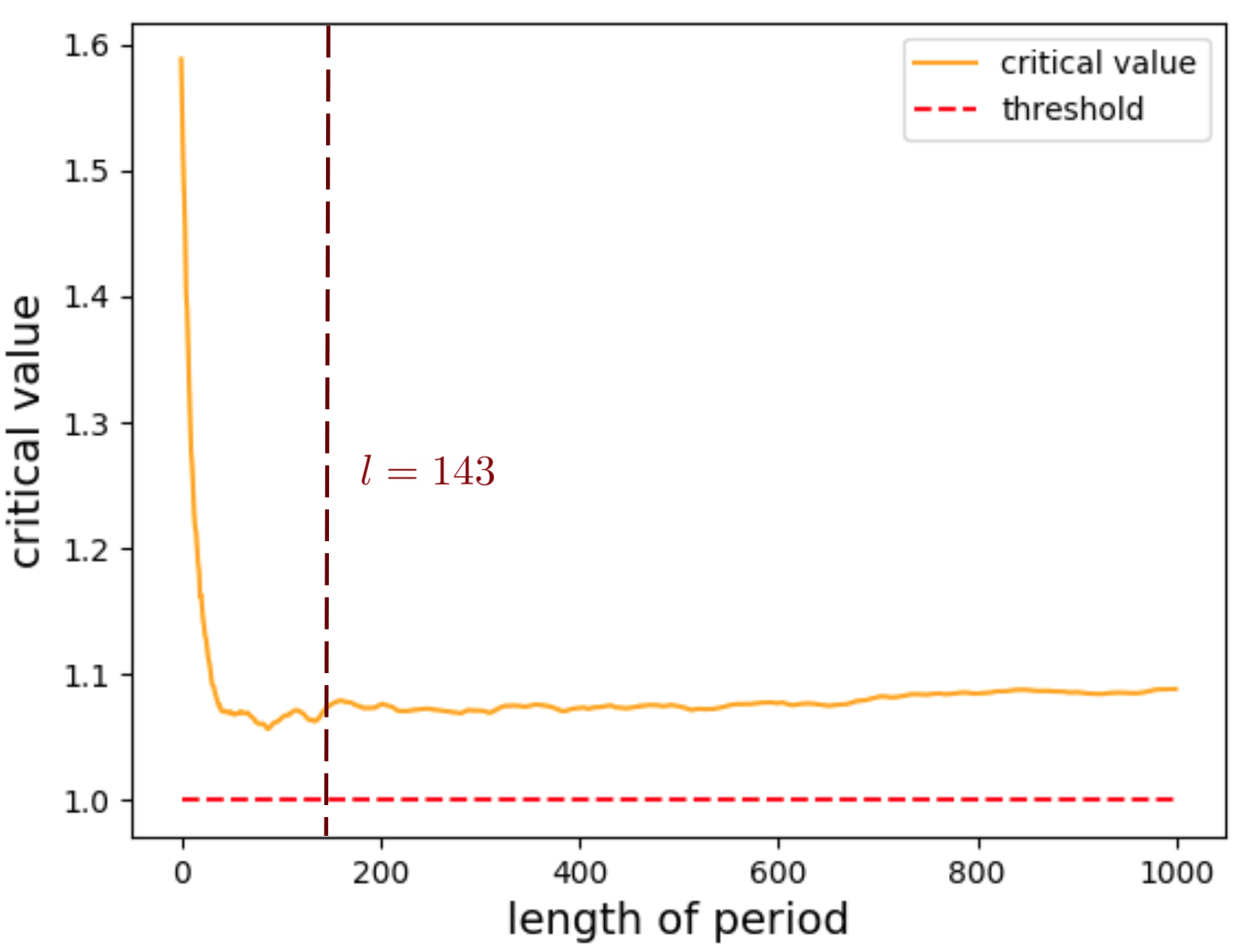}%
  \label{fig:choice-of-period-2}
}
\caption{The critical values $\rho_{\text{cr}}(\textbf{S})$ that we compute for different time periods $l$ for SIS processes. In each case, the SIS process occurs on a tie-decay network that we construct from the ER network $G^{(2)}_{\text{ER}}$ with a decay coefficient of $\alpha = 10^{-1}$ and interactions that we generate using an exponential waiting-time distribution with scale $\beta = 100$.}
\label{fig:choice-of-period}
\end{figure}

Because we compute the system matrix only for the first $l$ time steps, instead of for the entire time span, the fast convergence of critical values also allows us to save computation time when computing estimates of the critical values of the epidemic-threshold condition. Let $\rho_{\text{cr}}(\textbf{S}^{(l)})$ denote the estimated critical value that we compute when the period is $l$. For the numerical experiments in \Cref{subsec:validation} and \Cref{subsec:influence_factors}, we estimate the critical value for each period $l$ until we satisfy the following stopping criterion: $\| \max_{k \in \{l-9, \ldots, l\}}\rho_{\text{cr}}(\textbf{S}^{(k)}) - \min_{k \in \{ l-9, \ldots, l\}}\rho_{\text{cr}}(\textbf{S}^{(k)})\| \leq 0.02$. We usually finish this computation of critical values within about 100--200 time steps, which is much smaller than the $10^3$ time steps when we conduct numerical simulations of SIS dynamics over the entire time span.


\subsection{Experiments on Real-World Examples}
\label{subsec:real-world}

We now construct tie-decay networks using data from real-world examples and explore the dynamics of our SIS model on these networks. 
We consider two real-world data sets: (1) a workplace network of interactions between individuals in an office building in France between 24 June and 3 July in 2013 \cite{Genois2015} and (2) a conference network of face-to-face contacts over 2.5 days between conference attendees during the ACM Hypertext 2009 conference \cite{Isella2011}. For each data set, we use the time stamps of the interactions between people when we construct its tie-decay network. Given a data set, we initialize the state of a tie-decay network as follows: (1) we use the nodes that are present in the data set; (2) an edge exists between each pair of distinct nodes with an independent, homogeneous probability of $0.1$ (i.e., we create a network from the ensemble $\mathcal{G}(N, 0.1)$, where $N$ is the number of nodes), and we assign an initial tie strength of $0.5$ to each edge that exists. For each data set, we consider only a single initial network. We use the time stamps from the empirical data for the interactions and hence to determine the evolution of the tie strengths. The ties decay exponentially with a decay coefficient of $\alpha = 10^{-2}$, and we increment the tie strengths whenever an interaction takes place. As before, we validate our theoretical results using numerical simulations of SIS dynamics, and we also examine the influence of different choices of the time period $l$ on our computational estimates of the epidemic thresholds. 

In \Cref{fig:workplace-validation} and \Cref{fig:ht09-validation}, we compare the final outbreak sizes and estimated critical values in the workplace network and the conference network, respectively. These two real-world examples are both fairly small; the workplace network has 93 nodes, and the conference network has 113 nodes. The interactions in the workplace network have a roughly periodic pattern, with individuals interacting more frequently during work hours than during other hours. The conference network (which also was used in Valdano et al.\cite{valdano2015} to validate their epidemic threshold) has a different interaction pattern---for example, some individuals are in sequences of interactions during a short period of time, but then have few or no further interactions---than the workplace network because of the nature of a scientific conference. Despite the differences between the two real-world examples, we observe a strong correlation between the estimated critical values and the final outbreak sizes in both of them. Although the epidemic-threshold condition \eqref{eqn:epidemic_threshold} does not explicitly state that a larger critical value corresponds to a larger number of nodes in the infected state at $t = T$ (when we finish our simulations), this positive correlation tends to hold for both real-world networks. As in \Cref{subsec:validation}, we highlight the $(\lambda, \mu)$ pairs that yield critical values that are closest to $1$ and we indicate their corresponding outbreak sizes. In both real-world examples, whenever the critical values fall below $1$, the disease dies out at the end of a simulation. This supports our theoretical formulation of the epidemic-threshold condition in \eqref{eqn:epidemic_threshold}.

\begin{figure}[htbp]
\centering
\subfloat[Outbreak sizes at the end of the simulations.]{%
  \includegraphics[width=0.8\textwidth]{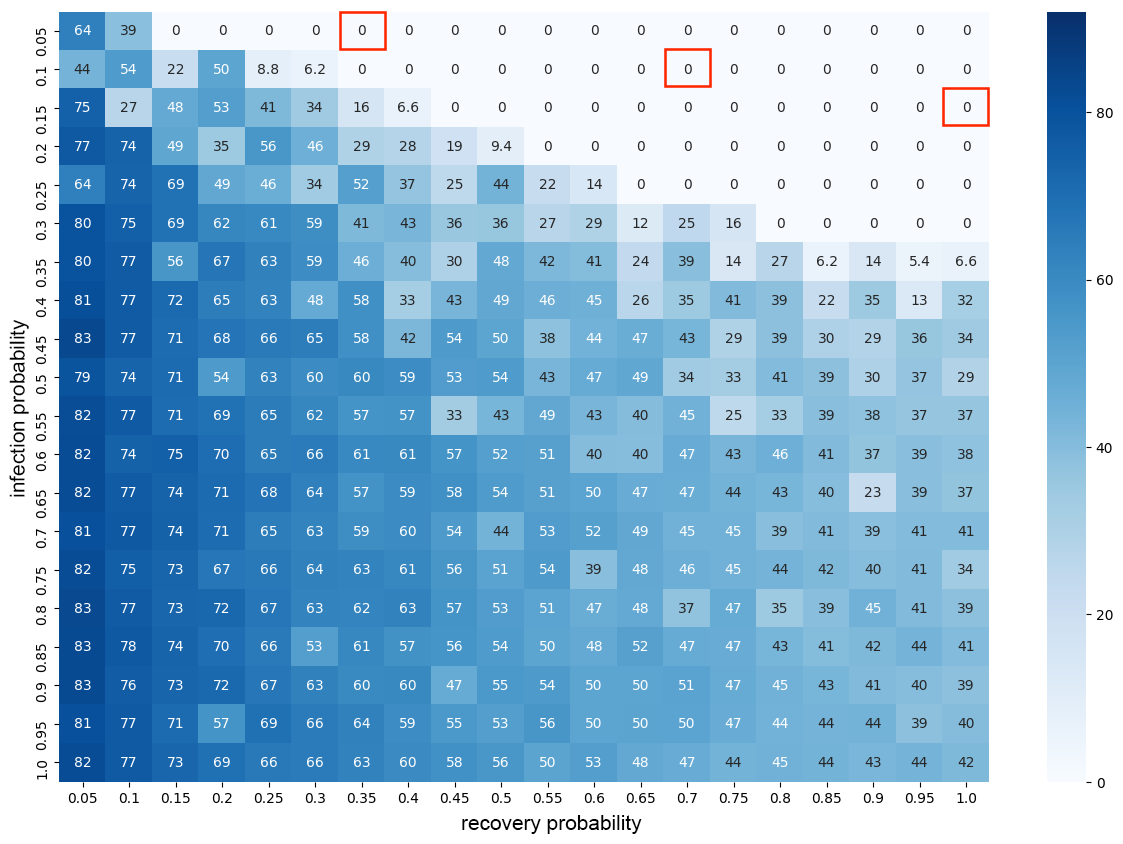}%
  \label{fig:workplace-validation-outbreak-size}
}\qquad
\subfloat[Critical values $\rho_{\text{cr}}(\textbf{S})$ of the epidemic-threshold condition.]{%
  \includegraphics[width=0.8\textwidth]{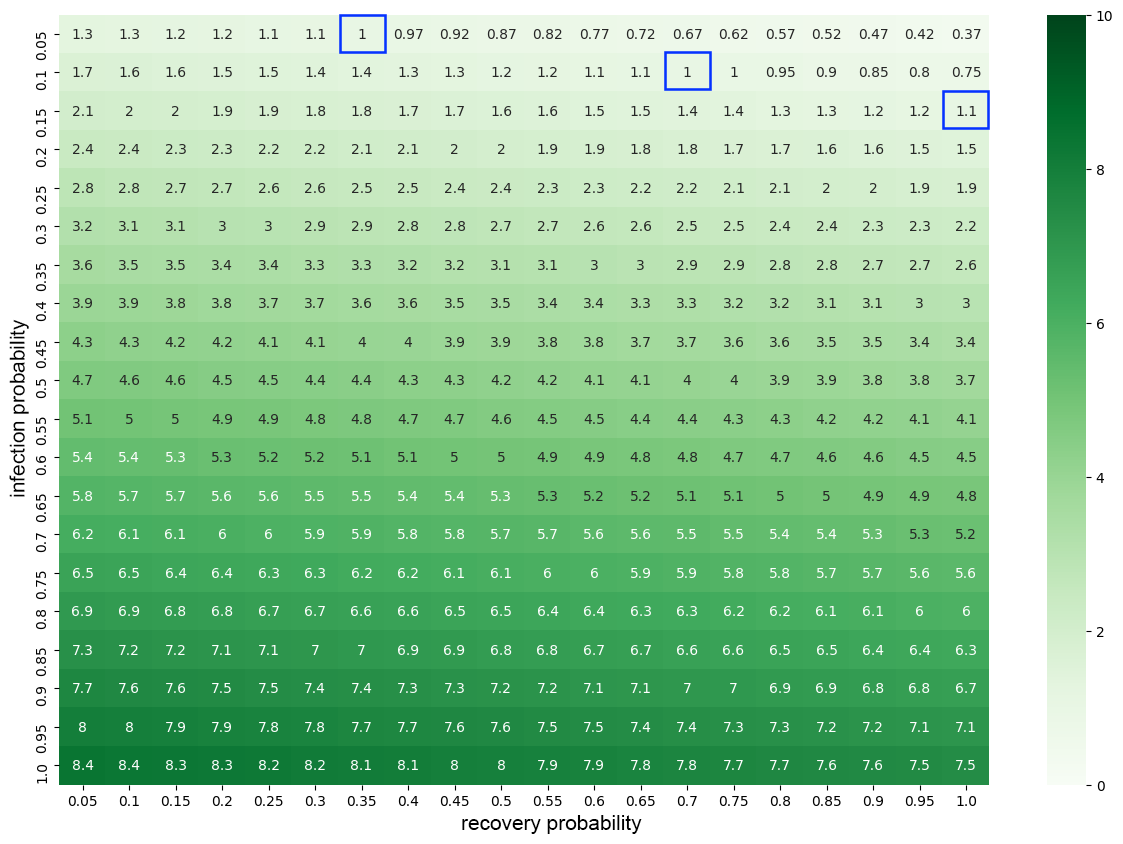}%
  \label{fig:workplace-validation-critical-value}
}
\caption{(a) Final outbreak sizes of our SIS process and (b) the associated estimated critical values of the epidemic-threshold condition for the workplace tie-decay network. We simulate the SIS dynamics for 988 time steps (where one time step consists of 1000 seconds) after discretization, and we estimate the critical values using a time period of length $l = 100$.
}
\label{fig:workplace-validation}
\end{figure}

\begin{figure}[htbp]
\centering
\subfloat[Outbreak sizes at the end of the simulations.]{%
  \includegraphics[width=0.8\textwidth]{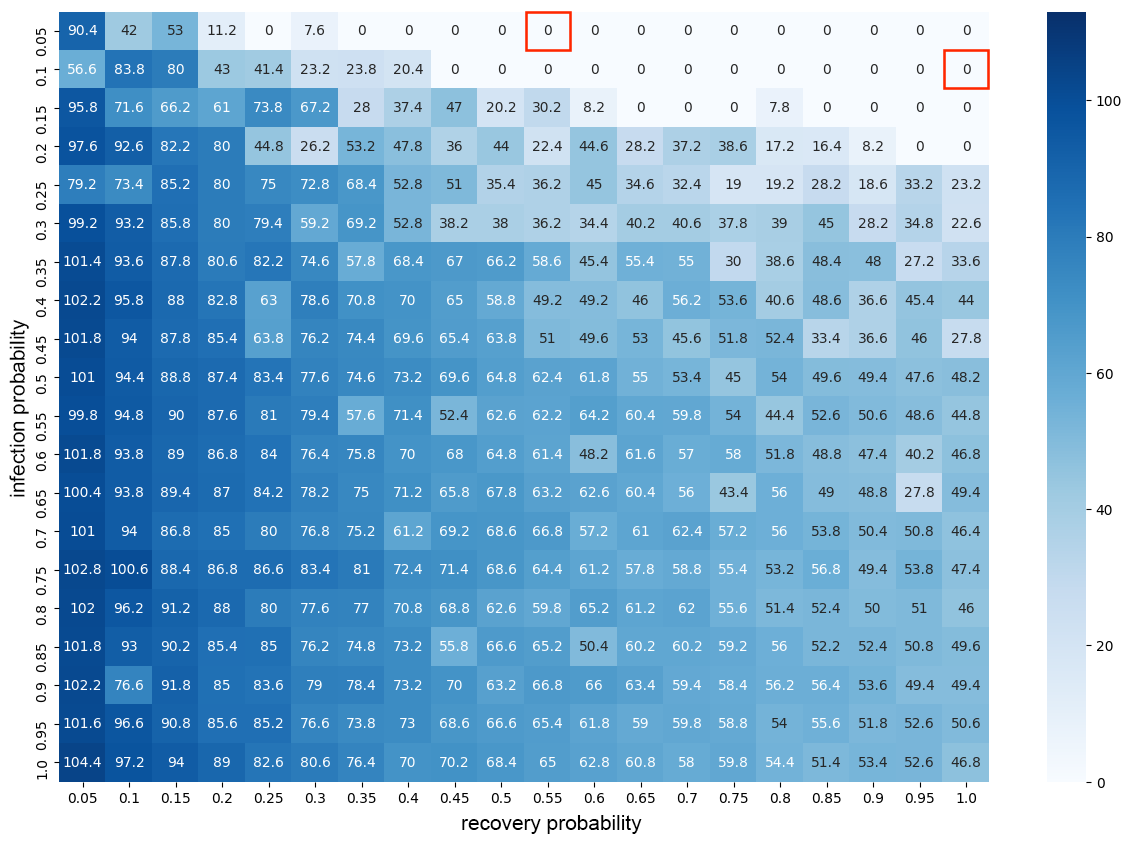}%
  \label{fig:ht09-validation-outbreak-size}
}\qquad
\subfloat[Critical values $\rho_{\text{cr}}(\textbf{S})$ of the epidemic-threshold condition.]{%
  \includegraphics[width=0.8\textwidth]{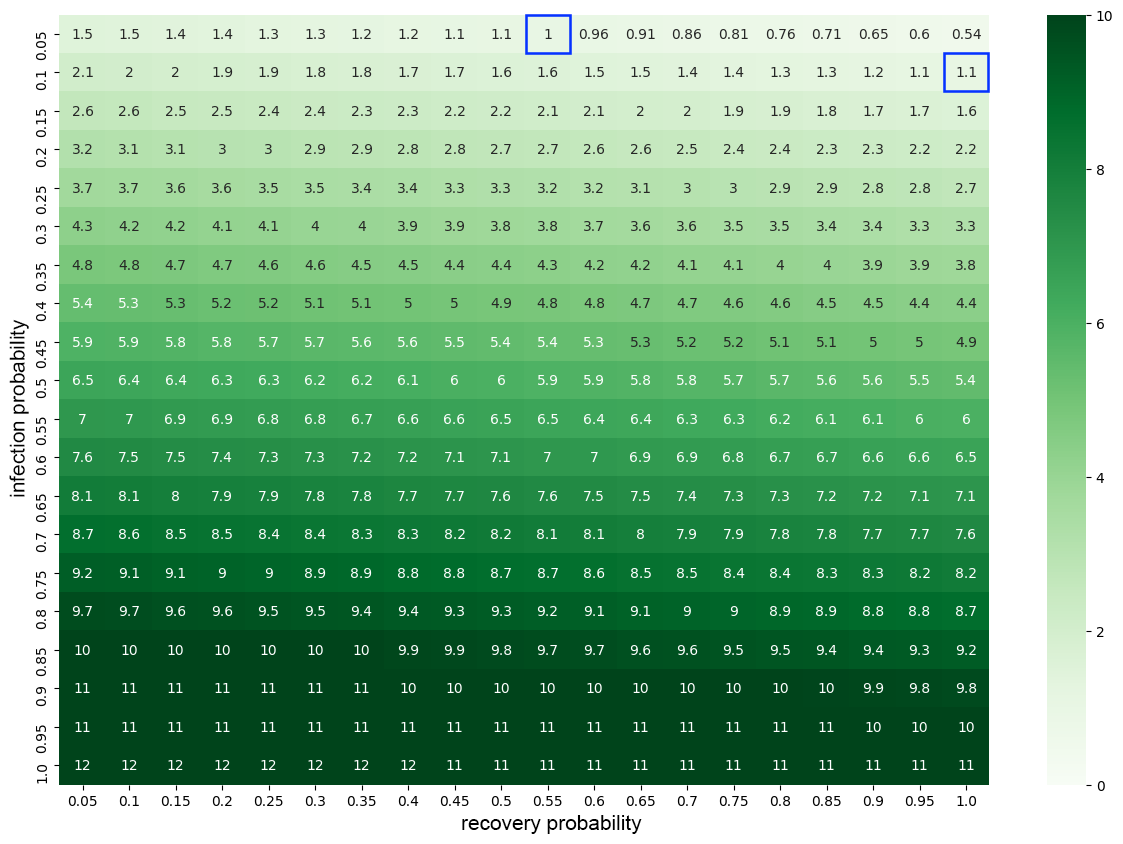}%
  \label{fig:ht09-validation-critical-value}
}
\caption{(a) Final outbreak sizes at the end of simulations of our SIS process and (b) the associated critical values of the epidemic-threshold condition for the conference tie-decay network. We simulate the SIS dynamics for 1,062 time steps (where one time step consists of 200 seconds) after discretization, and we estimate the critical values using a time period of length $l = 100$.}
\label{fig:ht09-validation}
\end{figure}

As we discussed in \Cref{sec:model}, to discretize time in the real-world networks, we choose a sufficiently small $\Delta t$ so that there are not too many interactions in the time interval $((\tau -1)\Delta t, \tau \Delta t]$. Specifically, for each of our real-world networks, we choose $\Delta t$ to ensure that the number of interactions in each time interval is no more than $10$. In the workplace data set, interactions were recorded every 20 seconds, and we take $\Delta t$ to be 1000 seconds in our discretization. Interactions were {also} recorded every 20 seconds in the conference data set, but this time we take $\Delta t$ to be 200 seconds in our discretization. We define one time step $\Delta t$ differently in the two data sets because of distinct features that we observe in their contact patterns. In the workplace network, there are many time intervals without any interactions, so we use a coarse discretization to ensure that the evolution of tie strengths is meaningful. In the conference network, interactions are more frequent, so we use a finer discretization. After discretization, each of the real-world networks has about 1,000 time steps in total. For each network, we estimate the epidemic threshold using approximately one tenth of the entire time span; specifically, we use a time period of $l = 100$. For the workplace network, this choice entails examining the critical value of the epidemic-threshold condition using all contacts from the first day; for the conference network, we use all contacts from the first 5.5 hours. Our discussion in \Cref{subsec:choice_of_period} suggests that data from the early stages of these temporal networks is sufficiently representative of the entire data set to allow us to successfully estimate the epidemic thresholds for the entire time span. Furthermore, for the workplace network, it is reasonable to assume that the contact patterns of the workers are somewhat periodic, with similar patterns during each work day. The contact patterns in the conference network also appear to have a somewhat periodic pattern, as there is a spike in the number of contacts approximately every 6 hours. In summary, for both networks, using the period $l = 100$ seems to give a good estimate of the epidemic threshold, as we observe a close relationship between the magnitudes of the critical values and the final outbreak sizes with this choice. 

Our accurate estimations of critical values of the epidemic-threshold condition using only early times in disease dynamics suggests the possibility of control measures for slowing down the spread of a disease. For instance, government regulations such as rules for physical distancing (which is also called ``social distancing'') can decrease the interaction frequencies of social contacts. As we saw in \Cref{fig:critical-val-temporal}, as we lower the interaction frequency $\beta$, the dividing line for the epidemic threshold shifts to the left. Therefore, for fixed infection and recovery probabilities, when the interaction frequency is sufficiently small, the critical value can become smaller than $1$, so a disease outbreak is unlikely. Additionally, the use of personal protective equipment (PPE) like masks can reduce infection probabilities, thereby also leading to a decrease of the critical value of the epidemic threshold. 
 

\subsection{Comparison with SIS Dynamics on a Traditional Temporal Network}
\label{subsec:comparison_traditional}

To highlight how features of tie-decay networks assist in the forecasting of epidemic outbreaks, we compare the epidemic thresholds that we obtain using a tie-decay network with ones that we obtain using a traditional temporal network that we construct from binning interactions in a time window. We also illustrate some challenges that arise if one simulates a model of disease spread on a network that aggregates all of the interactions into adjacent time windows of length $w$. This further motivates the use of tie-decay networks for studying spreading behavior on temporal networks.

To construct a traditional temporal network by binning interactions, we work with the ER network $G^{(1)}_{\text{ER}}$ and the same sequence of interactions (with time stamps $T_e = t_1, t_2, \ldots$ for each edge $e = (i, j)$) that we used in \Cref{subsec:validation}. We build a traditional temporal network using adjacent windows of length $w = 10$ \cite{Braha2009}. That is, we first divide the time span into adjacent, disjoint time windows $(10(k-1), 10k]$ and we then aggregate all of the interactions within each window. Let $\textbf{A}'_k$ denote the adjacency matrix of the $k$th window $(10(k-1), 10k]$. To make sure that it is reasonable to compare our results from using tie-decay networks with those from using traditional temporal networks, we also rescale the tie strengths of $\textbf{A}'_k$ such that their sum is equal to the time-averaged sum of tie strengths $\textbf{B}(t)$ within the $k$th time window. We then simulate an SIS process with a {maximum} infection probability of $\lambda$ and a recovery probability of $\mu$ on the traditional temporal network that consists of the sequence $\{\textbf{A}'_1, \textbf{A}'_2, \ldots \}$ of adjacency matrices. Within the $k$th window $(10(k-1), 10k]$, we simulate the SIS process for 10 steps; within this window, the tie strengths are constant and given by $\textbf{A}'_{k-1}$. Using methods that were designed for discrete temporal networks \cite{aditya2010,valdano2015}, we derive the epidemic-threshold condition for the traditional temporal network to be $\rho_{\text{cr}}(\textbf{S}') = 1$, where $\textbf{S}' \triangleq \prod_{k} \left[(1-\mu)\textbf{I} + \lambda_\text{max}\min\{\textbf{A}'_k , 1\}\right]$ is the system matrix that is associated with the traditional temporal network and $\rho_{\text{cr}}(\mathbf{\Theta})$ denotes the spectral radius of the matrix $\mathbf{\Theta}$. As in our terminology for tie-decay networks, we refer to $\rho_{\text{cr}}(\textbf{S}')$ as the ``critical value'' of the traditional temporal network.

In our comparison, we simulate an SIS process with different maximum infection probabilities $\lambda$ and a fixed recovery probability of $\mu = 0.5$ on the tie-decay network (see  \Cref{subsec:validation} for details of its properties) and the traditional temporal network that we construct from $G^{(1)}_{\text{ER}}$ and the same sequence of interactions. In \Cref{fig:comparison}, we plot their critical values and final outbreak sizes versus the recovery probability. One major difference between the dynamics on the two types of networks is the magnitudes of their critical values. The critical values $\rho_{\text{cr}}(\textbf{S})$ for the tie-decay network range from $0.64$ to $3.40$, whereas the critical values $\rho_{\text{cr}}(\textbf{S}')$ for the traditional temporal network range from $0.94$ to $1.06$ and remain close to $1$. The proximity of $\rho_{\text{cr}}(\textbf{S}')$ to the threshold value $1$ poses two challenges. First, although theoretical results \cite{aditya2010} suggest that, as time $t \rightarrow \infty$, one can successfully predict whether or not an outbreak will take place based on the epidemic-threshold condition, it may be difficult to obtain an accurate prediction in networks with finitely many nodes that one examines for only a finite amount of time. When the critical value is close to $1$, even a very small perturbation can change whether or not an epidemic-threshold condition is satisfied. In \Cref{fig:comparison}, the critical value of the traditional temporal network exceeds $1$ for $\lambda \geq 0.4$, but we observe outbreaks only for $\lambda \geq 0.7$. The second challenge is that the proximity of $\rho_{\text{cr}}(\textbf{S}')$ to $1$ also makes it difficult to discern the extent to which the critical values correlate with the final outbreak sizes. In \Cref{subsec:validation}, we examined the positive correlation between the critical value and final outbreak size for an SIS process on a tie-decay network. However, one can see in \Cref{fig:comparison} that such a correlation is less evident for an SIS process on a traditional temporal network.

\begin{figure}[htbp]
\centering
\includegraphics[width=0.9\textwidth]
{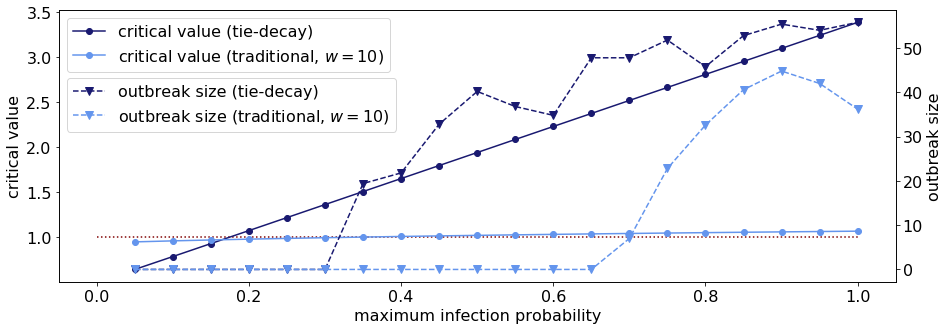}%
\caption{The critical values and final outbreak sizes that we obtain from simulating an SIS process on a tie-decay network and on a traditional temporal network with a time-window length $w=10$. We build the two networks using the same underlying ER network. (See the text for details.) The SIS process has a maximum infection probability of $\lambda$ and a fixed recovery probability of $\mu = 0.5$. In each network, we simulate the SIS process for $10^3$ time steps. We repeat each simulation $10$ times and report the means of the final outbreak sizes and critical values. We color the critical values and the final outbreak sizes for the tie-decay network in dark blue and those for the traditional temporal network in light blue. We mark the critical values with discs and the final outbreak sizes with triangles. The dotted red line marks the threshold value $1$.}
\label{fig:comparison}
\end{figure}

Another difficulty in constructing a traditional temporal network is determining an appropriate time-window length $w$ (or multiple such lengths, if one allows them to be nonuniform) \cite{psorakis2013probabilistic}. When one has prior knowledge of seasonality (or other regularity, such as periodicity) in data, it can be worthwhile to use a traditional network that is divided into a sequence of time windows. However, in many applications, such prior knowledge is typically not available. Modeling the spread of an infectious disease on a tie-decay network does not require tuning a time-window length, and it is thus worthwhile to investigate disease dynamics on tie-decay networks.


\section{Conclusions and Discussion}
\label{sec:conclusion}

We studied the epidemic threshold of an SIS process on tie-decay networks, which model relationships between nodes in a way that distinguishes between tie strengths and interactions between the nodes. In these tie-decay networks, the tie strengths increase instantaneously when there is an interaction and decay continuously in time between interactions. We demonstrated how to mathematically formulate an SIS process on a tie-decay network and then derived the epidemic threshold of this process by extending methods that were designed for networks that consist of sequences of temporal snapshots. Based on our theoretical results, we performed numerical simulations on both synthetic and real-world networks to obtain several insights into SIS dynamics on tie-decay networks. First, we showed numerically that the epidemic-threshold condition is successful at estimating the final outbreak sizes in the numerical simulations. We also showed that the critical value of the epidemic threshold is positively correlated with the final outbreak size of a disease. Our numerical experiments on synthetic networks illustrated how various factors---the decay coefficient of the tie strengths, the interaction frequency between nodes, and the sparsity of a network---impact the spread of a disease on a tie-decay network. Our numerical experiments on the length of the time period over which we computationally estimate the epidemic threshold demonstrated the possibility of estimating the critical values of disease dynamics using data from the early stages of disease spread. Finally, we demonstrated that one can estimate the epidemic threshold successfully in tie-decay networks that one constructs from real-world contact data.

There are a variety of interesting ways to build on our work. When deriving the epidemic threshold of our SIS model on a tie-decay network, we first discretized the network using a sufficiently small time step and we then applied methods that were designed for discrete-time temporal networks. It is also important to extend approaches for deriving epidemic thresholds that were designed for continuous-time temporal networks (see \cite{Speidel2017, valdano2018continuous}). Although the existing approaches to do this do not appear to be immediately applicable to tie-decay networks (because one cannot necessarily assume that the adjacency matrix at any time $t$ commutes with the aggregated matrix up to time $t$ due to the particular structure of tie-decay networks), it should be possible to modify them to incorporate the features of tie-decay networks. Another worthwhile research direction is to study epidemic thresholds in more complicated epidemic models, such as SEIR processes (and models of disease spread with many more compartments), on tie-decay networks. It is valuable to examine new approaches on simplistic models such as SIS processes and SIR processes, but realistic models of disease dynamics are typically more complicated \cite{Brauer2019}. When studying such models, it will be especially interesting to examine whether or not it is still possible to accurately estimate critical values of disease dynamics at early stages of disease spread. It is also relevant to compare disease dynamics on tie-decay networks to disease dynamics in different types of continuous-time network models (such as Hawkes processes \cite{laub2015,zipkin2016}) that integrate a point process with a network of interacting entities. Because of the self-exciting properties of a Hawkes process, it produces interactions that cluster in time. Prior studies have illustrated that such burstiness in contact patterns impacts epidemic-threshold conditions \cite{Zino2018}, so it will be interesting to investigate how to incorporate such point-process models into a tie-decay framework. Researchers continue to develop new types of temporal networks, and it is important to compare disease dynamics on tie-decay networks to such dynamics on these temporal networks. For example, as in the tie-decay networks that we employed, Gelardi et al.~\cite{gelardi2021temporal} recently examined temporal network data in the form of evolving weighted networks with edge weights that update from each interaction. However, unlike in our tie-decay networks, they took interconnections between social relationships into account. For example, an interaction between two individuals may simultaneously strengthen their relationship with each other while weakening their relationships with other individuals. It is important to explore how such interdependencies affect disease dynamics and other spreading processes.


\section*{Acknowledgements}

We thank Eugenio Valdano for helpful discussions.





\end{document}